\documentclass[5p]{elsarticle_1}

\usepackage{graphicx}      
\usepackage{natbib}        
\usepackage{graphicx}
\usepackage{graphics} 
\usepackage{epsfig} 
\usepackage{mathptmx} 
\usepackage{times} 
\usepackage{amsmath} 
\usepackage{amssymb}  
\usepackage{enumerate}
\usepackage{mathrsfs}
\usepackage{multirow}
\usepackage{subfigure}
\usepackage{latexsym}
\usepackage{bm}
\usepackage{dsfont}
\usepackage{multicol}
\usepackage{multirow}
\usepackage{color}
\usepackage{url}
\usepackage{algorithmic}
\usepackage{algorithm}
\usepackage{array}
\usepackage[T1]{fontenc}
\usepackage{comment}
{
	\begin{bmatrix}}%
	{\end{bmatrix}
	}

\newtheorem{remark}{Remark}
\newtheorem{lemma}{Lemma}

\newtheorem{proof}{Proof}

\newtheorem{theorem}{Theorem}

\newtheorem{assumption}{Assumption}

\graphicspath{{Figures/}}

\begin{document}
\begin{frontmatter}

\title{A hierarchical Lyapunov-based cascade adaptive control scheme for lower-limb exoskeleton}


\author[First,Fourth]{Xinglong Zhang}
\author[First,Second]{Wei Jiang}
\author[Third]{Zhizhong Li}
\author[Fourth]{Shengli Song}

\address[First]{College of Intelligence Science, National University of Defense Technology, Changsha 410073, China (zhangxinglong18@nudt.edu.cn)}
\address[Second]{Research Institute for National Defense Engineering of Academy of Military Science, Luoyang, 471300, P.R. China (weijiang@nudt.edu.cn)}
\address[Third]{State Key Laboratory of Disaster Prevention \& Mitigation of Explosion \& Impact, College of Defense Engineering, Army Engineering University, Nanjing, Jiangsu 210007, China (lizz0607@163.com)}
\address[Fourth]{College of Field Engineering, Army Engineering University, Nanjing, 210007, P.R. China (shengli.song66@gmail.com)}

\begin{abstract}
	 This paper proposes a hierarchical Lyapunov-based adaptive cascade control scheme for a lower-limb exoskeleton with control saturation. The proposed approach is composed by two control levels with cascade structure. At the higher layer of the structure, a Lyapunov-based back-stepping regulator including adaptive estimation of  uncertain parameters and friction force is designed for the leg dynamics, to minimize the deviation of the joint position and its reference value. At the lower layer, a Lyapunov-based neural network adaptive controller is  in charge of computing control action for the hydraulic servo system, to follow the force reference computed at the high level, also to compensate for model uncertainty, nonlinearity, and control saturation. \\
	 The proposed approach shows to be  capable in minimizing the interaction torque between machine and human, and suitable for possible imprecise models. The robustness of the closed-loop system is discussed under input constraint. Simulation experiments are reported, which shows that the proposed scheme is effective in imposing smaller interaction torque with respect to PD controller, and in control of models with uncertainty and nonlinearity.
\end{abstract}


\end{frontmatter}

\section{Introduction}
In recent years, persisting efforts have been devoted to the developments of lower-limb exoskeletons with the goal of alleviating body burden and  augmenting human motion performance  in the areas of military purposes and/or industrial applications. In these scenarios,  the payloads carried by the operators are usually quite heavy  that robotics with high power supplies are expected to reduce dramatically the loads applied to the operators.  Many exoskeletons in this direction have been developed, see~\cite{anam2012active,bogue2015robotic,yan2015review,aliman2017design,jimenez2012review} and the references therein. Among them, hydraulic actuators are usually used due to their large values of power/mass ratio, see for instance~\cite{ouyang2016development,cao2010design,nandor2012design}. Classic hydraulic actuated exoskeleton examples include Berkeley Lower Extremity Exoskeleton (BLEEX) \cite{zoss2006biomechanical,kazerooni2005control}, its updated version Human Universal Load Carrier (HULC), and etc.  Lower-limb exoskeleton system is highly nonlinear and in principle is difficult to obtain precise model parameters, which makes the control of such systems challenging.
Many works of literature address the control problems in this respect. For instance, in \cite{racine2003control}, several first-order sliding mode controllers have been proposed for the force tracking control of the hydraulic actuator of the exoskeleton. A Radius Basis Function (RBF)
based sliding mode control scheme has been developed in \cite {song2014rbf} to alleviate the chattering effect caused by the first-order sliding mode control. A simplified Lyapunov-based control approach has been addressed in \cite{alleyne2000simplified} for force tracking control of electro-hydraulic systems. Adaptation algorithm is proposed to estimate the model uncertain parameters. In~\cite{kim2017design}, a locomotive control algorithm has been addressed for normal stable walking with lower-limb exoskeleton actuated by hydraulic system. A cascade interaction torque control for hydraulic actuated lower-limb exoskeleton is proposed in~\cite{chen2016cascade}, where the integral of interaction torque is minimized to generate the joint trajectory to be followed via a PID controller. A tracking control algorithm for a knee exoskeleton has been developed in~\cite{li2017tracking}, where the interaction torque is considered as an unknown disturbance to be rejected. In~\cite{wang2016data}, a data-driven adaptive sliding mode control algorithm has been proposed for a multi degree-of-freedom robotic exoskeleton.\\
\begin{figure}[h!]
	\begin{center}
		\includegraphics[width=6.4cm]{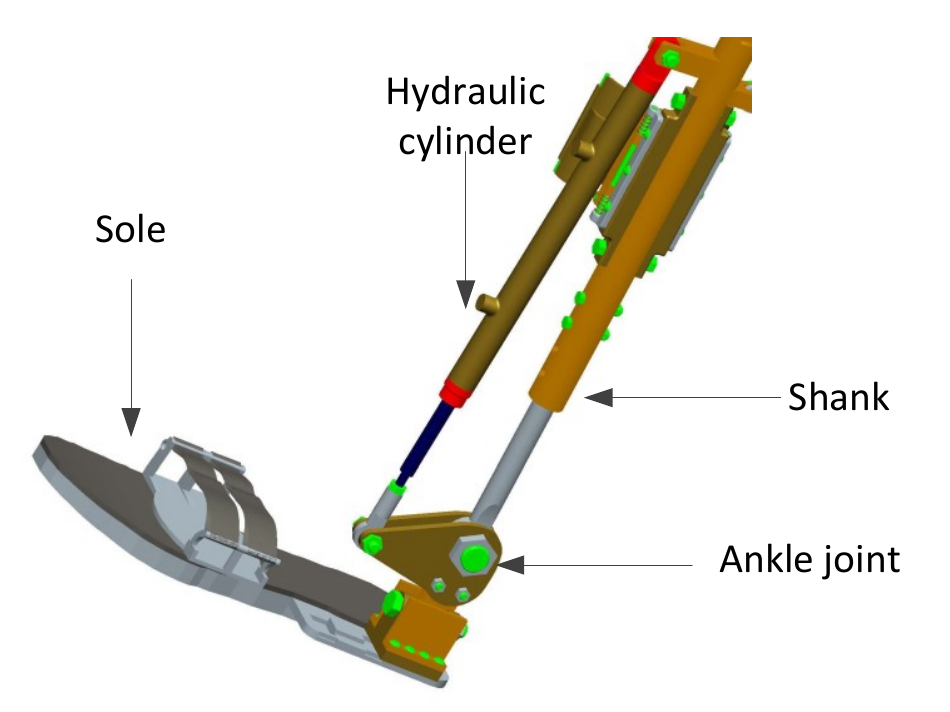}    
		\caption{Structure of ankle and shank. The flexion/extension freedom is driven by a cylinder.}
		\label{fig:leg}
	\end{center}
\end{figure}
Neural networks  have been found to be very useful for estimating nonlinear functions due to its powerful approximation capabilities. For this reason, neural networks are widely used in the context of adaptive control. An adaptive control of a class of nonlinear systems has been presented in \cite{ge2002direct} and the robustness of the closed-loop system has been proven. In \cite{zhao2016adaptive}, neural network based adaptive control has been extended to the active suspension system with actuator saturation.  Adaptive neural control of nonlinear systems with nonsmooth actuator nonlinearity has been considered in \cite{zhou2007adaptive} and a stable neural network observer has been developed in \cite{abdollahi2006stable}. However, most of the works prescribed need the assumption that the function to be approximated by neural network is continuous. In this work, with resorting to the techniques described in ~\cite{selmic1997neural,llanas2008constructive}, we extend the discussion to approximation and compensation for piecewise discontinuous function within the proposed control framework.\\
Another challenge of control of hydraulic actuated exoskeleton is how to compensate for the input saturation caused by the hydraulic actuator, which in principle might lead to poor tracking performance, including longer period of transient, unacceptable overshoot, even larger tracking error. To deal with this problem, many methods have been proposed, see for instance~\cite{chen2011adaptive,he2016adaptive}. Among them, an auxiliary system has been introduced in~\cite{he2016adaptive} for the impedance control of a robotic manipulator subject to input saturation. In~\cite{huang2013adaptive}, an adaptive control of mobile robot with torque saturation has been studied. The
torque has been designed to be a function constructed by model parameter that can be guaranteed within the saturation limit. In~\cite{xu2004iterative}, an iterative learning control approach has been proposed for nonlinear uncertain systems, of which the convergence of the state has been proven under control saturation. A PID controller has been used to control robot manipulators under bounded torque saturation with exponential stability property guaranteed using singular perturbation theory, see~\cite{santibanez2010practical}. Inspired by above techniques, in this work a Lyapunov-based virtual system is introduced  to compensate for the control saturation, where its model terms and parameters are properly designed according to Lyapunov direct method.  \\
 Fig.~\ref{fig:leg} depicts a left ankle joint equipped with a hydraulic actuator. The control problem of such a lower-limb exoskeleton is addressed in this work. Note that, as the control of one joint can be easily extended to that of the whole exoskeleton system, for simplicity, only the control problem of an ankle joint is considered in this work.\\
To be specific, a two-layer Lyapunov-based adaptive cascade control scheme for joint position control of a lower-limb exoskeleton is presented. Its control diagram is depicted in Figure~\ref{fig:scheme}. At the higher layer of the control structure, a Lyapunov-based adaptive regulator is designed for the leg dynamics including adaptation algorithms for uncertain parameters and friction estimations, with the goal of minimizing the deviation of the joint position $\varphi$ and its reference value $\varphi_d$. The outcome of this layer is the desired value of the hydraulic actuator force reference $F_{L,d}$ to be followed at the low level. At the lower layer, a Lyapunov-based neural network adaptive regulator computes the input signal $u$ for the hydraulic exoskeleton system with the scope of tracking the desired force $F_{L,d}$, meanwhile compensates for the unknown time-varying parameter and piecewise discontinuous nonlinearity with an integrated neural network composed by both continuous and discontinuous basis functions, and for control saturation with an auxiliary virtual system. The proposed approach shows to be capable in minimizing the interaction torque between machine and human, and suitable for imprecise models.
\begin{figure}[h!]
	\begin{center}
		\includegraphics[width=8.4cm]{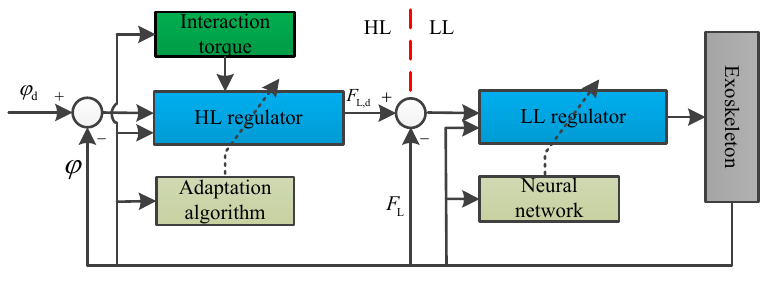}    
		\caption{ Diagram of the proposed control scheme. HL=Higher layer, LL=Lower layer.}
		\label{fig:scheme}
	\end{center}
\end{figure}
The overall robustness property of the two levels is analyzed under input constraint. Simulation experiments including a comparison with PD controller are reported, so as to verify the effectiveness of the proposed approach in reference tracking and in interaction torque minimization, and its potentiality for control of models with uncertainty and nonlinearity. \\
The main contributions of this work are summarized as follows:
\begin{itemize}
	\item
	The adopted cascade structure can improve the control performance with respect to the single layer control scheme when there exists model uncertainty and/or exogenous disturbance affecting directly the lower layer system. In this case,  the proposed algorithm can indeed limit the effects caused by the uncertainty and/or disturbance at the lower layer on the controlled variables at the higher layer. 	
	\item The approach allows for multi-objectives at the high layer, e.g., the minimization of the joint position tracking error and the optimization of the interaction torque between machine and human.
	\item  The synthesis of continuous RBF and discontinuous jump approximate function are used to estimate piecewise discontinuous function caused by system nonlinearity and friction force.
	\item The proposed control scheme shows to be suitable for systems with possibly multiple unknown (possibly time-varying) parameters and nonlinear functions.
\end{itemize}

The rest of the paper is organized as follows. In Section ~\ref{sec:2} the model is described and the control goal is introduced. The design of the high-level and low-level Lyapunov-based adaptive controller is presented in Sections~\ref{sec:3} and~\ref{sec:4}, respectively, while the closed-loop properties  are described in Section~\ref{sec:5}. The simulation results are reported  in Section \ref{sec:6}, while some conclusions are drawn in Section \ref{sec:7}.

\section{System description and control objective}\label{sec:2}
In this section the nonlinear model of the overall system
under study and the main idea to deal with the control problem are described.
\subsection{System description}
The ankle joint of the exoskeleton is driven by a hydraulic servo system, and its schematic diagram is depicted in Fig.~\ref{fig:hydraulic}. The operating principle of the hydraulic system is as follows. When the relay is closed, high-pressure oil from the pump flows through the check valve into the accumulator in case the gas pressure of the accumulator is smaller than the oil pressure. Meanwhile, a second flow direction is made where the oil flows through the left hole of the four-sided servo valve into the left (right) chamber of the valve core. The oil then flows through the servo valve to the actuator in order to drive an external load. Notice that, the measured force  of the cylinder from the force sensor is transmitted to the regulator such that the control action applied to the actuator is computed according to the error of the desired force and the measured one.
Valve openings are changeable according to the input signal computed with the controller. This determines the values of pressure and flow in the acting chamber of the actuator.
\begin{figure}
	\begin{center}
		\includegraphics[width=8.4cm]{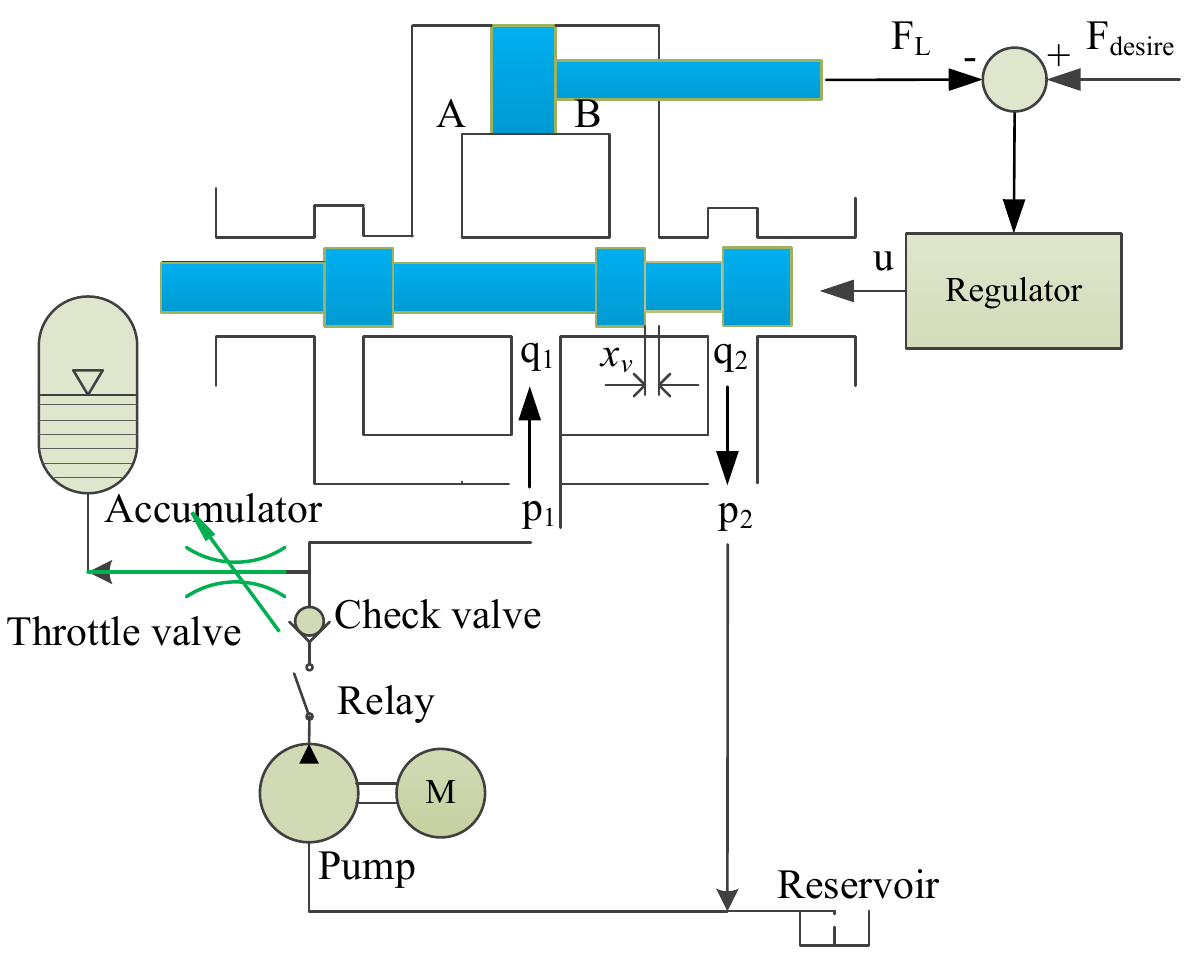}    
		\caption{The hydraulic system contains: a reservoir, pump, throttle valve, relay,
			check valve, accumulator, four-sided servo valve and cylinder.}
		\label{fig:hydraulic}
	\end{center}
\end{figure}
It is highlighted that the accumulator is designed to supply oil in place of the pump when the relay is open, with the objective of saving energy and prolonging the working time. The throttle valve is selected to guarantee the flow coming from (going into) the accumulator remaining almost constant.\\ 
To sum up, there are two working modes for the hydraulic system.
\begin{itemize}
	\item Mode 1: The relay switches on. In this case, the pump provides oil for the system, in the meanwhile, charges the accumulator with oil. Once the pressure of the accumulator reaches its high threshold value, the relay switches off, and the pump stops working.
\item Mode 2: The relay switches off. In this case, the accumulator persistently releases energy so as to provide oil to the system. Once the pressure of accumulator reaches its low threshold value (minimum working limitation value), the relay switches on, and the whole system goes back to Mode 1.
\end{itemize}
The oil flow-rate of the four-sided servo valve is given by
\begin{equation}\label{Eqn:ql}
q_{L}=K_qx_v-K_cP_L
\end{equation}
where $K_q$, $K_c$, and $x_v$ are the idle flow gain coefficient, the flow-pressure coefficient and the spool position of servo valve respectively; while $P_L$  is the hydraulic pressure associated with the external load.
Under the assumption that the compressibility of the fluid is zero, the flow balancing equation in the cylinder is simplified as
\begin{equation}\label{Eqn:ql_ba}
q_{L}=\frac{1}{2}(q_1+q_2)
\end{equation}
where $q_1$ and $q_2$ are the flow rates into chambers A and B respectively.\\
 Furthermore, it also holds that
 \begin{subequations}\label{Eqn:qi}
 	\begin{align}
 q_1=& {C_{in}}(P_1-P_2) + {C_{ec}}{P_1} +\dot {\mathcal{V}}_1 +\frac{{{{\mathcal{V}}_1}}}{\beta }{\dot P_1}\\
 q_2=& {C_{in}}(P_1-P_2) - {C_{ec}}{P_2} -\dot {\mathcal{V}}_2-\frac{{{{\mathcal{V}}_2}}}{\beta }{\dot P_2}
 \end{align}
 \end{subequations}
where $A_1$, $P_1$,  and ${\mathcal{V}}_1$ are the area of piston, the hydraulic pressure and the volume in chamber A, while $A_2$, $P_2$, and ${\mathcal{V}}_2$ are the corresponding counterpart in chamber B. $\beta$ is the effective bulk modulus.   $C_{in}$ and $C_{ex}$ are the cylinder internal and external leakage coefficients. $\dot P_i$ and $\dot {\mathcal{V}}_i$ are the derivatives of  $P_i$ and ${\mathcal{V}}_i$, $i=1,2$.\\
The  hydraulic pressures of chamber A and B satisfy
 \begin{subequations}\label{Eqn:PL-P12}
    \begin{align}
	P_L=&P_1-P_2\\
P_s=&P_1+P_2
	\end{align}
\end{subequations}
where $P_s$ is the outlet oil pressure of the pump and $P_L$ is the pressure associated with the external load.
Assuming that $P_s$ is derivable, from \eqref{Eqn:PL-P12}, it holds that
\begin{subequations}\label{Eqn:dp}
	\begin{align}
\dot P_1=&\frac{\dot P_L+\dot P_s}{2}\\
\dot P_2=&\frac{\dot P_L-\dot P_s}{2}
\end{align}
\end{subequations}
The volume ${\mathcal{V}}_i$, $i=1,2$, respect the geometric principle, i.e.,
\begin{subequations}\label{Eqn:V}
	\begin{align}
	{\mathcal{V}}_1=&{\mathcal{V}}_0+A_1x_c\\
	{\mathcal{V}}_2=&{\mathcal{V}}_0-A_2x_c
	\end{align}
\end{subequations}
where $x_c$ is the piston position of the cylinder.\\
In view of \eqref{Eqn:ql}-\eqref{Eqn:V}, the cylinder flow reads
\begin{equation}\label{Eqn:ql_r}
{q_L} = \frac{{{A_1} + {A_2}}}{2}{\dot x_c} + ({{C_{in}} + {C_{ec}}}){P_L} +\frac{{{{\mathcal{V}}_0}}}{{2\beta }}{\dot P_L}
\end{equation}
where $\dot x_c$ is the piston velocity of the cylinder.\\
The mathematical model of the accumulator is presented as
 \begin{equation}\label{Eqn:accum}
{P_0}{\mathcal{V}}_0^{r_0} = {P_h}{\mathcal{V}}_h^{r_0} = {P_l}{\mathcal{V}}_l^{r_0}
\end{equation}
where $r_0>1$ is a constant value,  $P_h$ and $P_l$ are the corresponding high and low threshold values of pressure in the accumulator. The high threshold value is chosen as $P_h=P_p$, where $P_p$ is the outlet oil pressure of the pump; ${\mathcal{V}}_h$ and ${\mathcal{V}}_l$ are the associated gas volume.
 Assuming that the relay switches off at a generic time instant $t$, the future pressure in the accumulator drops till it reaches the low threshold value, see Fig.~\ref{fig:accupre}. Thanks to the use of throttle valve, the inlet (outlet) flow of the accumulator can be assumed to be constant, then the pressure satisfies
 \begin{equation}\label{Eqn:accum_pre}
P_{t+\Delta t} = P_h\left( {\frac{{{{\mathcal{V}}_h}}}{{{{\mathcal{V}}_h} + {q_a} \cdot \Delta t}}} \right)^{r_0}
\end{equation}
where  $q_a$ is the inlet (outlet) flow of the accumulator.
Therefore, the working pressure during the two modes can be represented as
\begin{subequations}\label{Eqn:prew}
	\begin{align}
{P_s} = \left\{ \begin{array}{ll}
	{P_p},& \text{Mode1}\\
	{P_{t+\Delta t}},& \text{Mode2}
	\end{array} \right.
		\end{align}
Moreover, its derivative is defined as
	\begin{align}\label{Eqn:dot_prew}
{\dot P_s} = \left\{ \begin{array}{ll}
0,& \text{Mode1}\\
{\dot P_{t+\Delta t}},& \text{Mode2}
\end{array} \right.
\end{align}
\end{subequations}
\begin{figure}[h!t!]
	\begin{center}
		\includegraphics[width=8.4cm]{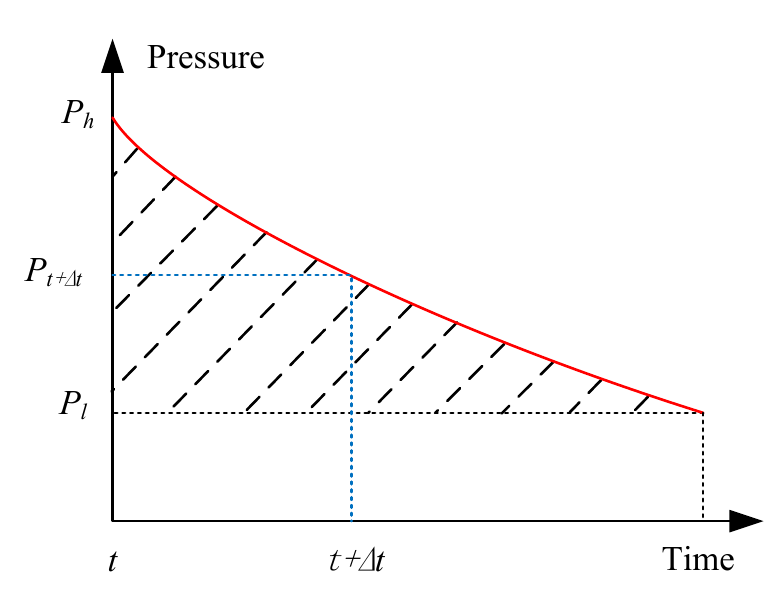}    
		\caption{The relationship between the pressure and the working time $\Delta t$  of the accumulator.}
		\label{fig:accupre}
	\end{center}
\end{figure}
By inspecting the motion function of the actuator, the force that drives the external load is given in the form of
\begin{equation}\label{Eqn:FL}
F_L=A_1P_1-A_2P_2
\end{equation}
Recalling the structure of the leg (see again Fig.~\ref{fig:leg}),  the output force $F_L$ of the hydraulic servo system acts as the control input for the inverted pendulum model of the leg dynamics, which is represented as (see ~\cite{racine2003control})
 \begin{equation}\label{Eqn:leg-dy}
J\ddot \varphi +mgr\sin (\varphi )= N({F_L} + {F_f})+\tau_{hm}
 \end{equation}
 \begin{figure}[h!t!]
 	\begin{center}
 		\includegraphics[width=8.4cm]{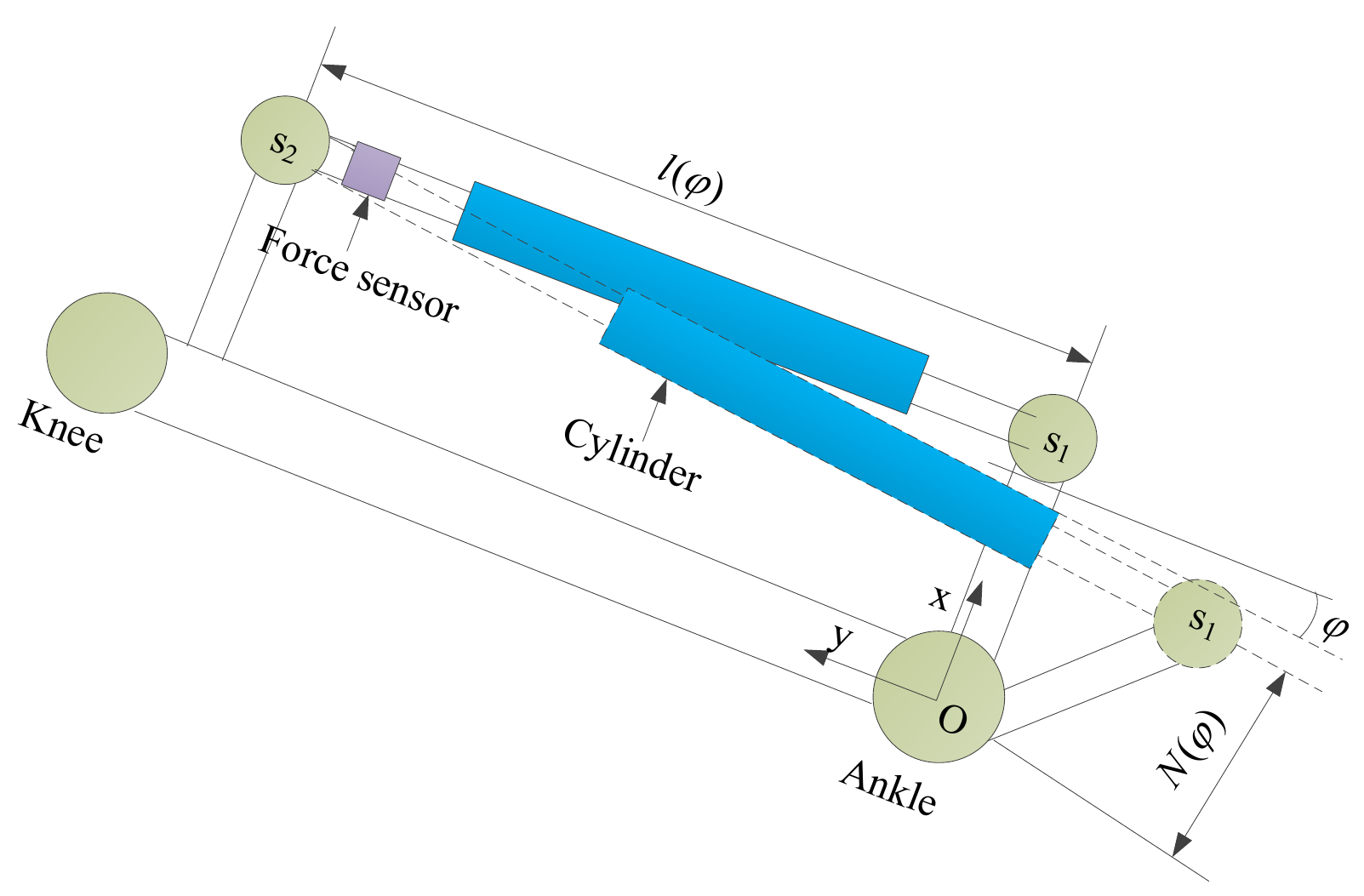}    
 		\caption{ Kinematical structure of the actuated shank.}
 		\label{fig:structure}
 	\end{center}
 \end{figure}
where $\varphi$ is the joint angle of the ankle, $J$ is the inertia of the shank, $N$ is the actuator moment arm (see Fig.~\ref{fig:structure}), $m$ is the mass of the shank, $r$ is the center position of the mass of the shank, $F_f$ is the friction force of the piston. $\tau_{hm}$ is the interaction torque between machine and human described by
\begin{equation}\label{Eqn:tau_hm}
\tau_{hm}=k_{p}(\varphi-\varphi_d)+k_{d}(\dot\varphi-\dot\varphi_d)
\end{equation}
where parameters $k_p$ and $k_d$ are constant values that amplify the differences of the angles and their velocities between machine and human.
The joint angle of the ankle establishes a relationship with the position of the piston, that is
\begin{equation}\label{Eqn:xc}
x_c=l(\varphi)-l_0-x_{c0}
\end{equation}
where $l_0$ is the initial length of the cylinder, $x_{c0}$ is the initial position of the piston, $l(\varphi)$ is the length of the cylinder that can be calculated according to the geometrical analysis of the shank as shown in Fig.~\ref{fig:structure}.
To this end, first note that the cylinder is placed between the ankle joint and the knee joint mounting at the points $s_1$ and $s_2$, and $O$ is the origin of the coordinates for the ankle joint. Denoting by $r_{s1}$ and $r_{s2}$ the distance from $s_1$ and $s_2$ to the coordinate origin, the length of the cylinder $l(\varphi)$ can be computed according to Law of cosines, that is
\begin{equation}\label{Eqn:l}
l(\varphi)=\sqrt{-2r_{s1}r_{s2}cos(\varphi-\theta_1-\theta_2)+r_{s1}^2+r_{s2}^2}
\end{equation}
 where $\theta_1=tan^{-1}(a_1/-b_1)$, $\theta_2=tan^{-1}(a_2/b_2)$, $a_1$, $a_2$ are the coordinate positions of $s_1$ and $s_2$ along the $x$ axis, while $b_1$ and $b_2$ are the coordinate positions  of $s_1$ and $s_2$ along the $y$ axis  respectively.
 The moment arm in \eqref{Eqn:leg-dy} is a function of $l(\varphi)$:
 \begin{equation}\label{Eqn:mom_arm}
N=r_{d1}\text{sin}(\text{cos}^{-1}(\frac{r_{d2}^2-l(\varphi)^2-r_{d1}^2}{-2l(\varphi)r_{d1}}))
 \end{equation}
The friction force in the hydraulic cylinder is not negligible and can be composed by Coulomb force, Viscous force, Stribeck effects and position dependent forces, see \cite{astolfi2007nonlinear,chantranuwathana2004adaptive}. However, it has been noticed that the Stribeck effects and the position dependent forces are usually very small, thus for simplicity, they are neglected in this case. The friction model of the piston is given by
 \begin{equation}\label{Eqn:friction}
 {F_f} = -F_C \operatorname{sgn} ({\dot \varphi}) - b{\dot \varphi}
 \end{equation}
 where $F_C$ and $b$ are the Coulomb friction term and the viscous friction coefficient respectively. \\
 The input to the hydraulic servo system is the electrical current, that affects directly the spool position through the mechanical time constant $\tau$ and the DC gain of the valve current $k_s$ to the spool position $x_v$. This relationship is described as
\begin{equation}\label{Eqn:xv}
\dot x_v=\frac{1}{\tau}(k_su-x_v)
\end{equation}
It is usually unavoidable in servo valve that the input is saturated by its maximum and minimum limits, see for instance \cite{yao2000adaptive}. The control saturation of the actuator can be described by
 \begin{equation}\label{Eqn:saturation}
 u= \left\{ \begin{array}{lc}
 u_{max}&u> u_{max}\\
 u&u_{min}\leq u\leq u_{max}\\
 u_{min} &u< u_{min}
 \end{array} \right .
 \end{equation}
 where $u_{min}\leq u_{max}$.
 If the computed value of the control variable   violates the constraint $\begin{bmatrix}
 u_{min}&u_{max}
 \end{bmatrix}$, control saturation occurs and the input is partially applied to \eqref{Eqn:xv}. The residual input that can not be implemented is defined as $\delta$ and is represented by \cite{chen2011adaptive}
 \begin{equation}\label{Eqn:resigual_input}
 \delta= \left\{ \begin{array}{lc}
 u_{max}-u&u> u_{max}\\
 0&u_{min}\leq u\leq u_{max}\\
 u_{min}-u &u< u_{min}
 \end{array} \right .
 \end{equation}
In view of \eqref{Eqn:ql_r}-\eqref{Eqn:xv}, one can write the system in the state-space form
\begin{equation}\label{Eqn:state}
\Sigma: \left\{ \begin{array}{l}
\ddot \varphi  = \frac{1}{J}(N({F_L} + {F_f}) - mgr\sin (\varphi )+\tau_{hm}) \\
{\dot F}_L = {n_1}{x_v} - {n_2}{{\dot x}_c} - {n_3}{F_L} + {n_4}{P_s}+n_5\dot P_s\\
{{\dot x}_v} = \frac{1}{\tau }({k_s}u - {x_v})
\end{array}  \right.
\end{equation}
where $${n_1} = \frac{{2\beta ({A_1} + {A_2}){K_q}}}{{{2{\mathcal{V}}_0+(A_1-A_2)x_c}}},\ {n_2} = \frac{{\beta {{({A_1} + {A_2})}^2}}}{{2{\mathcal{V}}_0+(A_1-A_2)x_c}},$$ $${n_3} = \frac{{2\beta (2{K_c} + {2C_{in}} + {C_{ex}})}}{{2{\mathcal{V}}_0+(A_1-A_2)x_c}},\ {n_4} =  \frac{{\beta (2{K_c} + 2{C_{in}} + {C_{ex}})({A_2} - {A_1})}}{{2{{\mathcal{V}}_0}+(A_1-A_2)x_c}},$$ $$n_5=\frac{A_2-A_1}{2}-\frac{{ {{({A_1} + {A_2})}^2}}}{{2(2{\mathcal{V}}_0+(A_1-A_2)x_c)}}.$$
\subsection{Control objective}
Considering a reference trajectory $\varphi_d$ for the angle variable of the ankle joint  and denoting $e_1$ the deviation of the value of the real angle and the desired one, i.e.,
\begin{equation}\label{Eqn:e_1}
	e_1=\varphi-\varphi_d
\end{equation}
the control goal for system \eqref{Eqn:state} is to drive $e_1$ to the origin.
 The control problem is trivial if the  system \eqref{Eqn:state}  is precise. In fact, the system \eqref{Eqn:state} is a simplified representation especially for the hydraulic servo model. Moreover, the mass $m$ and the inertia $J$ of the inverted pendulum model \eqref{Eqn:leg-dy} are usually difficult to be accurately measured. For these reasons, adaptive control algorithms  that estimate the constant and time-varying parameters, and compensates for the friction force, neglected nonlinearity, and time-varying pressure supply, can be used so as to achieve the  aforementioned control goal.
Note also that, the model \eqref{Eqn:state} can be partitioned into a hierarchical structure, where the high-level dynamic is
 \begin{subequations}\label{Eqn:hierar-model}
 \begin{align}\label{Eqn:high_m}
 \Sigma_H:
 \ddot \varphi  = \frac{1}{J}(N({F_L} + {F_f}) - mgr\sin (\varphi )+\tau_{hm}),
 \end{align}
while the low-level model is given as
  \begin{align}\label{Eqn:low_m}
 \Sigma_L: \left\{ \begin{array}{l}
 {\dot F}_L = {n_1}{x_v} - {n_2}{{\dot x}_c} - {n_3}{F_L} + {n_4}{P_s}+n_5\dot P_s\\
 {{\dot x}_v} = \frac{1}{\tau }({k_s}u - {x_v})
 \end{array}  \right.
 \end{align}
 \end{subequations}
{\color{black}
It can be noted that, \eqref{Eqn:hierar-model} coincides with the cascade model structure in which, the variable $F_L$ acts as the input to $\Sigma_H$ and as one of the outputs to $\Sigma_{L}$, establishing a direct interconnection between the two subsystems.  As previously described, in principle, a single-layer adaptive controller can be designed regarding~\eqref{Eqn:hierar-model} as a whole. However, in this way, the closed-loop control performance might be sensitive to model uncertainty and disturbance. In view of this, the proposed solution in this note is an adaptive Lyapunov-based approach with hierarchical cascade control structure. This control approach naturally corresponds to the model structure, with which, the influence of possible uncertainty and nonlinearity from $\Sigma_L$ to $\Sigma_H$ can be suppressed, e.g., via properly designed compensation algorithm at the low level. A brief description of the controller at the two levels is as follows.} The higher layer regulator is in charge of computing the required value of $F_L$ for \eqref{Eqn:high_m} such that the goal $E=e_1^{\top}e_1$ is minimized. The constant unknown parameters, e.g. the mass $m$, and the inertia $J$, are properly estimated, and the friction force is accounted for. At the lower layer, the regulator is designed for \eqref{Eqn:low_m} with the scope to track the force $F_L$ computed at the higher layer and to compensate for the possible uncertainty and nonlinearity.
The following standard assumption is assumed to be holding:
\begin{assumption}\label{assu:measure}
	The cylinder force $F_L$, ankle joint angle  $\varphi$, velocity $\dot \varphi$, and the accelerated velocity $\ddot \varphi$ are measurable.
\end{assumption}

\section{Design of the high-level Lyapunov-based adaptive controller}\label{sec:3}
In this section, the high-level Lyapunov-based controller is designed and the adaptive parameters estimation algorithms are introduced with the objective to drive $e_1$ to the origin.
\subsection{High-level Lyapunov-based controller}
First, denoting $x_1=\varphi$, $x_2=\dot\varphi$, rewrite the model \eqref{Eqn:high_m} as
 \begin{equation}\label{Eqn:high_m_x}
\Sigma_H: \left\{ \begin{array}{l}
\dot x_1=x_2\\
\dot x_2 = \frac{1}{J}(N({F_L} + {F_f}) - mgr\sin (x_1 )+\tau_{hm})
 \end{array}  \right.
\end{equation}
In view of \eqref{Eqn:e_1}, define a virtue control variable $v_1=-k_1e_1+\dot{\varphi_d}$, $k_1>0$ and denote \begin{equation}\label{eqn:e_2}
e_2=x_2-v_1=k_1e_1+(\dot\varphi-\dot\varphi_d)
\end{equation}
By applying the concept of back-stepping control (see \cite{chen1996backstepping}), we define a new system corresponding to \eqref{Eqn:high_m_x} in the following form:
 \begin{equation}\label{Eqn:e_3-b-model}
\Sigma_H: \left\{ \begin{array}{l}
\dot e_1=-k_1e_1+e_2\\
\dot e_2 = \frac{1}{J}(N({F_L} + {F_f}) - mgr\sin (x_1 )+\tau_{hm}) -\dot v_1
\end{array}  \right.
\end{equation}
In order to design a proper Lyapunov-based controller, as consider the following Lyapunov function
\begin{equation}\label{Lyap_1}
V_{H,1}=\frac{1}{2}\rho_1 e_1^2+\frac{1}{2}\rho_2Je_2^2
\end{equation}
where $\rho_1>0$, $\rho_2>0$ are the penalty weights associated with $e_1$ and $e_2$ respectively.
Taking the derivative of $V_{H,1}$ in \eqref{Lyap_1}, one has

\begin{equation*}
\begin{array}{rll}
&\dot V_{H,1}=\rho_1e_1(-k_1e_1+e_2)+\rho_2Je_2(\dot x_2-\dot v_1(t))\\
&=-k_{\rho_1}e_1^2-k_{\rho_2}e_2^2+\\
&\rho_2e_2\underbrace{(k_2e_2+\rho_{12}e_1+N(F_L+F_f)-mgrsin(x_1)+\tau_{hm}-J\dot v_1)}\\
&\ \ \ \ \ \ \ \ \ \ \ \ \ \ \ \ \ \ \ \ \ \ \ \ \ \ \ \ \ \ \ \ \ \ \ \ \ \ \ \
\ \ \ \ \ \text{set}=0
\end{array}
\end{equation*}
where $k_2>0$, $k_{\rho_1}=k_1\rho_1$, $k_{\rho_2}=k_2\rho_2$, and $\rho_{12}=\rho_1/\rho_2$.\\
Therefore, the control action can be selected as
\begin{equation}\label{Eqn:FL_I}
F_L=-\frac{1}{N}(k_2e_2+\rho_{12}e_1+NF_f+\tau_{hm}-mgrsin(x_1)-J\dot v_1)
\end{equation}
such that $\dot V_{H,1}=-k_{\rho_1}e_1^2-k_{\rho_2}e_2^2\leq0$.
where $k_{\rho_1}$ and $k_{\rho_2}$ are the tuning knobs taht define the decaying rate of the Lyapunov function $V_{H,1}$.
\begin{remark}
	In view of the definition of $\tau_{hm}$ in~\eqref{Eqn:leg-dy} and $e_2$ in~\eqref{eqn:e_2}, if parameter $k_1$ is selected such that $k_1=k_p/k_d$, one promptly has $e_2=\tau_{hm}/k_d$.  In this case, it is easy to see that
	the interaction torque $\tau_{hm}$ can also be minimized via the second term in the right-hand side of \eqref{Lyap_1}.
\end{remark}
\subsection{Adaptive design with the estimations of $m$ and $J$}\label{sec:m_J}
The input defined in \eqref{Eqn:FL_I} is highly model dependent  that  the mass $m$ and the inertia $J$ are assumed to be accurately measured which however in principle is nontrivial in practical situation. In view of this, an adaptive update algorithm is proposed to estimate $m$ and $J$. To this end, denote by $\hat m$ and $\hat J$ the estimated values of $m$ and $J$, and by $\tilde m=m-\hat m$ and $\tilde J=J-\hat J$ the corresponding estimation errors. Consider the following augmented Lyapunov function as
\begin{equation*}
V_{H,2}=V_{H,1}+\frac{1}{2}q_J\tilde{J}^2+\frac{1}{2}q_m\tilde{m}^2
\end{equation*}
where $q_J$ and $q_m$ are positive scalars.
A new control action in place of \eqref{Eqn:FL_I} is selected by substituting the corresponding estimated values for $m$ and $J$, i.e.,
\begin{equation}\label{Eqn:FL_I_s}
F_L=-\frac{1}{N}(k_2e_2+\rho_{12}e_1+NF_f+\tau_{hm}-\hat mgrsin(x_1)-\hat J\dot v_1)
\end{equation}
Taking the derivative of $V_{H,2}$ and applying the input \eqref{Eqn:FL_I_s}, one has
 \begin{equation*}
 \begin{array}{l}
\dot V_{H,2}=\\ =-k_{\rho_1}e_1^2-k_{\rho_2}e_2^2-\rho_2e_2(\tilde{m}grsin(x_1)+\tilde{J}\dot v_1)-q_J\tilde{J}\dot{\hat{J}}-q_m\tilde{m}\dot{\hat{m}}\\

=-k_{\rho_1}e_1^2-k_{\rho_2}e_2^2-\tilde{J}\underbrace{(\rho_2e_2\dot v_1+q_J\dot{\hat{J}})}-\tilde{m}\underbrace{(\rho_2e_2grsin(x_1)+q_m\dot{\hat{m}})}\\
\ \ \ \ \ \ \ \ \ \ \ \ \ \ \ \ \ \ \ \ \ \ \ \ \ \ \ \ \ \ \ \ \ \ \ \ \ \ \ \   \text{set}=0\ \ \ \ \ \ \ \ \ \ \ \ \  \ \ \ \ \ \ \ \ \ \ \ \ \text{set}=0
\end{array}
 \end{equation*}
 The parameter adaptation algorithm of $\hat J$ and $\hat m$ can be chosen as
 \begin{subequations}\label{Eqn:esti_update}
 	\begin{align}
 	\dot{\hat{J}}=&-\frac{1}{q_J}\rho_2e_2\dot v_1\\
 	\dot{\hat{m}}=&-\frac{1}{q_m}\rho_2e_2grsin(x_1)
 	\end{align}
 \end{subequations}
Thus, under the control choice \eqref{Eqn:FL_I_s} and the estimation updating rule \eqref{Eqn:esti_update}, it holds that
\begin{equation*}
\dot{V}_{H,2}=-k_{\rho_1}e_1^2-k_{\rho_2}e_2^2.
\end{equation*}
\subsection{Adaptive design with friction force estimation}
In the previous section~\ref{sec:m_J}, model \eqref{Eqn:friction} is included in the proposed control action for the objective of friction force compensation. Notice that, the parameters of the friction model can be obtained according to repeated experiments validation, see \cite{alleyne2000simplified}. In this work, instead, we propose an adaptation algorithm to estimate the parameters of friction force within the structure of \eqref{Eqn:friction}. In doing so, the cumbersome experiment validation can be avoided. To this end, define the corresponding estimated friction force as
 \begin{equation}\label{Eqn:friction_est}
{\hat F_f} = -{\hat F_C}\operatorname{sgn} ({\dot \varphi}) - \hat b{\dot \varphi}
\end{equation}
where $\hat F_C$ and $\hat b$ are the estimated values of  $F_C$ and $b$ respectively.
Considering the estimated friction force as one of the feedback terms, the control action is amended as
\begin{equation}\label{Eqn:FL_I_s_f}
F_L=-\frac{1}{N}(k_2e_2+\rho_{12}e_1+N\hat{F}_f+\tau_{hm}-\hat mgrsin(x_1)-\hat J\dot v_1)
\end{equation}
Accordingly, a new Lyapunov candidate is considered in the form
\begin{equation*}
V_{H,3}=V_{H,2}+\frac{1}{2}q_C\tilde F_C^2+\frac{1}{2}q_b\tilde b^2
\end{equation*}
where $\tilde F_C=F_C-\hat F_C$ and $\tilde b=b-\hat b$; $q_C$ and $q_b$ are positive scalars. \\
 Taking the derivative of $V_{H,3}$ and applying the control action~\eqref{Eqn:FL_I_s_f}, it holds that
 \begin{equation*}
\dot V_{H,3}=-k_{\rho_1}e_1^2-k_{\rho_2}e_2^2+\rho_2Ne_2(F_f-\hat F_f)-q_C\tilde F_C\dot{\hat F}_C-q_b\tilde b\dot{\hat b}
 \end{equation*}
In view of \eqref{Eqn:friction} and \eqref{Eqn:friction_est}, one has
 \begin{equation*}\begin{array}{ll}
\dot V_{H,3}=&-k_{\rho_1}e_1^2-k_{\rho_2}e_2^2-\rho_2Ne_2\tilde F_C\operatorname{sgn} ({\dot \varphi}) -\\ &\rho_2Ne_2\tilde b{\dot \varphi}
-q_C\tilde F_C\dot{\hat F}_C-q_b\tilde b\dot{\hat b}
\end{array}
\end{equation*}
By using the associative law of addition in polynomials, then one has
\begin{equation}\label{Eqn:lyap3}
\begin{array}{lll}
\dot V_{H,3}=&-k_{\rho_1}e_1^2-k_{\rho_2}e_2^2-\tilde F_C\underbrace{(\rho_2Ne_2\operatorname{sgn}{\dot \varphi}+q_C\dot{\hat F}_C)}-\\
&\ \ \ \ \ \ \ \ \ \ \ \ \ \ \ \ \ \ \ \ \ \ \ \ \ \ \ \ \ \ \ \ \ \ \ \ \ \ \ \ \ \ \ \ \ \text{set}=0\\
&\tilde b\underbrace{(\rho_2Ne_2\dot \varphi+q_b\dot{\hat b})}.\\
&\ \ \ \ \ \ \ \ \ \ \text{set}=0
\end{array}
\end{equation}
Therefore, the adaptation algorithm for parameters $\hat F_C$ and $\hat b$ are selected as follows:
\begin{subequations}\label{Eqn:fric_para_adapt}
	\begin{align}
	\dot{\hat F}_C=&-\frac{1}{q_C}\rho_2Ne_2\operatorname{sgn}{\dot \varphi}\label{Eqn:FC}\\
	\dot{\hat b}=&-\frac{1}{q_b}\rho_2Ne_2\dot \varphi
	\end{align}
\end{subequations}
{\color{black}
\begin{remark}
It is worth mentioning that, as the right-hand sides of the adaptation algorithm~\eqref{Eqn:FC} and $\dot e_2$ with~\eqref{Eqn:e_3-b-model} are piecewise affine on the state variables, the uniqueness solution condition of the resultant closed-loop system might not be verified, see~\cite{lasalle1960some}. Therefore, the usage of the Lyapunov stability theorem, such as LaSalle's theorem, for the stability analysis of~\eqref{Eqn:e_3-b-model} is not straightforward. Nevertheless, we show in the following that the corresponding convergence property can  be proven in the Filippov sense, see~\cite{shevitz1994lyapunov,dieci2011sliding}.
\end{remark}
\begin{theorem}\label{theo:1}
	With the control action~\eqref{Eqn:FL_I_s_f}, and the adaptation algorithms~\eqref{Eqn:esti_update} and~\eqref{Eqn:fric_para_adapt},  the closed-loop form of~\eqref{Eqn:e_3-b-model} is asympotically stable, that the tracking error $e_1,\,e_2$ converge to the origin as time goes to infinity.
\end{theorem}
\begin{proof}
	With~\eqref{Eqn:FL_I_s_f},~\eqref{Eqn:esti_update}, and~\eqref{Eqn:fric_para_adapt}, it is possible to write the closed-loop form, of which we highlight that, in addition to that of~\eqref{Eqn:FC}, the right-hand side of $\dot e$ is also piecewise affine on the state variables due to~\eqref{Eqn:FL_I_s_f}. To show the convergence of $e_1$, $e_2$, we have to check the monotonicity of the Lyapunov function $V_{H,3}$ for all the Filippov solutions. To proceed, we split the state variables into two regions: $\mathcal{Z}_1=\{z\in\mathbb{R}^6|e_2<-\dot \varphi_d\}$,  $\mathcal{Z}_2=\{z\in\mathbb{R}^6|e_2>-\varphi_d\}$, and a separating surface between $\mathcal{Z}_1$ and $\mathcal{Z}_2$, that is $\mathcal{Z}_3=\{z\in\mathbb{R}^6|e_2=-\varphi_d\}$, where $z=\begin{bmatrix}
	e_1&e_2&\hat J&\hat m&\hat F_C&\hat b
	\end{bmatrix}^{\top}$. In view of this, we first show that, for $z\in \mathcal{Z}_1,\, or\, \mathcal{Z}_2$, the term $sgn(\dot\varphi)$ can be replaced by constant values, which verifies the smoothness of the closed-loop system. Thus, the derivative of $V_{H,3}$ can be easily computed, i.e., $\dot V_{H,3}=-k_{\rho_1}e_1^2-k_{\rho_2}e_2^2$. For the case $z\in\mathcal{Z}_3$, i.e., when the sliding motion might occurs, $sgn(\dot \varphi)=0$ leads to $\hat F_C=\emph{constant}$, to the smoothness of the closed-loop system, and to $\dot V_{H,3}=-k_{\rho_1}e_1^2-k_{\rho_2}e_2^2$. To sum up, for almost all $t$, and $\forall \, z$,
	\begin{equation}\label{Eqn:lyap3_d}
	\dot V_{H,3}=-k_{\rho_1}e_1^2-k_{\rho_2}e_2^2,
	\end{equation}
	which implies that  $\dot V_{H,3}<0$ for $\begin{bmatrix}
		e_1&e_2
	\end{bmatrix}^{\top}\neq 0$ along all the Filippov solutions, thus the origin is globally stable in the Filippov sense.\\
	Also, according to~\cite{orlov2003switched}, from~\eqref{Eqn:lyap3_d}, it holds that, for all $t$, 	\begin{equation}\label{Eqn:lyap3_d-bound} V_{H,3}\leq \bar V,
	\end{equation}
	where $\bar V$ is an abitrarily large finite positive scalar.
	In view of~\eqref{Eqn:lyap3_d} and,~\eqref{Eqn:lyap3_d-bound}, recalling that $V_{H,3}$ is differentiable and the right-hand sidde term  $-k_{\rho_1}e_1^2-k_{\rho_2}e_2^2$ is smooth, by Barbalat's Lemma, the variables $e_1$ and $e_2$ converge to the origin respectively as time goes to infinity. \hfill $\square$
\end{proof}
}
 In view of Theorem~\ref{theo:1},  the control goal is achieved, i.e., the angle $\varphi$ converges to its desired value $\varphi_d$. Furthermore, in view of the definition of $e_2$, it follows that the angle velocity $\dot{\varphi}$ converges to $\dot \varphi_d$. Note that, it is not guaranteed that the estimated value $\hat m$, $\hat J$, $\hat F_C$,  and $\hat b$ will converge to their true values. It is because their true values are usually unknown, thus the differences with respect to their estimated values are not available at all times and can not guarantee to be exactly compensated via adding feedback error terms in the corresponding adaptation laws. Even so, this fact does not make any negative influence on the convergence of $e_1$ and $e_2$.
The desired trajectory of $F_L$ to be tracked by the low-level controller is given as
\begin{equation}\label{Eqn:FL_d}
F_{L,d}=-\frac{1}{N}(k_2e_2-\rho_{12}e_1+N\hat{F}_f+\tau_{hm}-\hat mgrsin(x_1)-\hat J\dot v_1)
\end{equation}
 $F_{L,d}$ is a reasonable choice for the reference signal of the low layer for the reason that $e_1$ and $e_2$ can always be measured in real-time and the unknown parameters are replaced with estimated ones that are adaptively updated by the proper design at the high layer.

\section{Devise of the low-level Lyapunov-based neural network adaptive regulator}\label{sec:4}
In this section,  the low-level Lyapunov-based  regulator with the neural-network based estimation algorithm  is designed for \eqref{Eqn:low_m}, with the objective to track the reference signal $ F_{L,d}$ computed at the high layer, and to compensate for unknown time-varying parameters and nonlinear piecewise discontinuous function.
\subsection{Low-level Lyapunov-based regulator}
 First note that the friction force model contains the $\operatorname{sgn}(\dot \varphi)$ term, thus it is not derivable in the domain  $\dot \varphi=0$.  For this reason, we define by $\dot {\hat F}_{f}$ an approximated derivative function  associated with $ {\hat F}_{f}$, i.e.,
\begin{equation}\label{Eqn:F_f_d}
\dot{\hat F}_{f}= \left \{ \begin{array}{ll}
\dot {\hat F}_f^+, &\dot \varphi\geq \epsilon_0\\
\frac{-\hat F_C-\hat b\epsilon_0}{\epsilon_0}, &-\epsilon_0\leq\dot \varphi<\epsilon_0\\
		\dot {\hat F}_f^-,& \dot \varphi< \epsilon_0
	\end{array}\right .
\end{equation}
where $\dot {\hat F}_f^+=-\hat F_C-\hat b\ddot \varphi-\dot {\hat b} \dot \varphi$, $\dot {\hat F}_f^-=\hat F_C-\hat b\ddot \varphi-\dot{\hat b} \dot \varphi$, $\epsilon_0$ is a small positive scalar, $\ddot \varphi$ is the second derivative of $\varphi$. 
%
\begin{assumption}\label{rem:f_hat}
  Assume that $\epsilon_0$ is chosen small enough such that the friction force $\hat F_f$ can be properly approximated by $\int^{t}\dot {\hat {F}}_f$, i.e.,  $\hat F_f=\int^{t}\dot {\hat F}_f$.

\end{assumption}
%
In order to derive the model to be used at the low level,
 we denote $x_3=F_L$,  and rewrite  the system \eqref{Eqn:low_m}  as
 \begin{equation}\label{Eqn:low_m_3}
\left\{ \begin{array}{l}
\dot x_3=n_1x_v-n_2\dot x_c-n_3 x_3+n_4P_s+n_5\dot P_s\\
{{\dot x}_v} = \frac{1}{\tau }({k_s}u - {x_v})
\end{array}  \right.
\end{equation}
Thanks to the definition of $\dot{\hat F}_{f}$, we define by $e_3$ the deviation of $x_3$ and the desired one, i.e., $e_3=x_3-F_{L,d}$, then it is possible to write \eqref{Eqn:low_m_3} as
 \begin{equation}\label{Eqn:low_m_3_e}
\left\{ \begin{array}{l}
\dot e_3=n_1(x_v+f_4)\\
{{\dot x}_v} = \frac{1}{\tau }({k_s}u - {x_v})
\end{array}  \right.
\end{equation}
where $f_4=\frac{1}{n_1}(-n_2\dot x_c-n_3x_3-\dot {\tilde F}_{L,d}-\dot{\hat F}_f+n_{14}P_s+n_{15}\dot P_s)$, $ {\tilde F}_{L,d}= {F}_{L,d}+\hat F_f$, $ n_{14}=\frac{n_4}{n_1}$, $ n_{15}=\frac{n_5}{n_1}$.
In the following, the aim is to design the Lyapunov-based controller for the low-level system \eqref{Eqn:low_m_3_e} with the objective of minimizing the deviation of $F_L$ and $F_{L,d}$. To this end,
 first consider the following Lyapunov candidate
\begin{equation*}
V_{L,1}=\frac{\rho_3}{2n_1}e_3^2
\end{equation*}
where $\rho_3>0$.
Taking the derivative of $\dot V_{L,1}$ leads to
\begin{equation}\label{Eqn:lypu_d_L1}
\dot V_{L,1}=-k_{\rho_3}e_3^2+\rho_3e_3(k_3e_3+x_v+f_4)
\end{equation}
where $k_3$ is chosen as a positive scalar, and $k_{\rho_3}=k_3\rho_3$.\\
As the input variable $u$ did not appear in \eqref{Eqn:lypu_d_L1}, the second derivative of $V_{L,1}$ might be needed if back-stepping method is used at this level. However, note that the term $\dot {\hat F}_f$ is not continuous by the definition in \eqref{Eqn:F_f_d}, moreover, $n_5 P_s$ is piecewise at the switching time instant from Mode 1 to Mode 2. For this reason, the back-stepping method is not advisable in this case.\\
To solve this problem, stabilizing $e_3$,  the spool position in~\eqref{Eqn:lypu_d_L1} has to be set equal to
\begin{equation}\label{Eqn:xv_input}
x_v=-(k_3e_3+f_4)
\end{equation}
 As the valve dynamics ${{\dot x}_v} = \frac{1}{\tau }({k_s}u - {x_v})$ is linear, stable, and fast, it is reasonable to assume that the input $u$ is proportional to the spool position $x_v$, i.e., $\dot x_v=0$, that is
\begin{equation}\label{Eqn:u_kv}
u=\frac{1}{k_s}x_v
\end{equation}
Substituting \eqref{Eqn:u_kv} into \eqref{Eqn:xv_input} gives the choice of the input variable, i.e.,
\begin{equation}\label{Eqn:lypu_d_L1_u}
u=-\frac{1}{k_s}(k_3e_3+f_4)
\end{equation}
Considering the real dynamics ${{\dot x}_v} = \frac{1}{\tau }({k_s}u - {x_v})$, with~\eqref{Eqn:lypu_d_L1_u},  \eqref{Eqn:lypu_d_L1} can be rewritten as
\begin{equation*}
\begin{array}{l}
\dot V_{L,1}=-k_{\rho_3}e_3^2+\rho_3e_3(k_3e_3+(k_su-\tau\dot x_v))
\end{array}
\end{equation*}
In view of \eqref{Eqn:lypu_d_L1_u}, one has
\begin{equation*}
\begin{array}{ll}
\dot V_{L,1}&=-k_{\rho_3}e_3^2-\rho_3e_3\tau \dot x_v
\end{array}
\end{equation*}
\subsection{Adaptive neural-network based  estimation of $f_4$}
Notice that the term $f_4$ in \eqref{Eqn:low_m_3_e} contains nonlinearity, discontinuity, and possible uncertain time-varying parameters (e.g., $n_1$, $n_2$, $n_3$), thus it is nontrivial to be accurately measured or estimated with adaptation algorithms similar to the ones developed in Section~\ref{sec:3}. To properly estimate and compensate for $f_4$, a multi-layer neural network including both continuous RBF function and jump approximation basis function is proposed in this work.To this scope, it is highlighted that  the discontinuity jump points of $\dot {\hat {F_f}}$ is at $\dot {\varphi}=\epsilon_0\, \text{and}\, -\epsilon_0$, while the ones of $P_s$ and $\dot P_s$ are at $P_s=P_l$. With above information, it is possible to write $f_4$ in terms of the proposed neural network in the following form:
\begin{equation}\label{Eqn:neural_network}
f_4=W_1^{\top}h(Z)+W_2^{\top}\phi(Z+c_1)+W_3^{\top}\phi(Z+c_2)+\epsilon_f(Z)
\end{equation}
where $c_1=\begin{bmatrix}
0& \epsilon_0& 0&P_l
\end{bmatrix}^{\top}$, $c_2=\begin{bmatrix}
0& -\epsilon_0& 0&P_l
\end{bmatrix}^{\top}$, $W_1\in \mathbb{R}^{n_1\times 1}$, $W_2\in \mathbb{R}^{n_2\times 1}$, and $W_3\in \mathbb{R}^{n_2\times 1}$, are the ideal output weight vectors, $n_1$ and $n_2$ are the corresponding numbers of the neurons, $Z=\begin{bmatrix}x_3&\dot \varphi&\dot{ \tilde F}_{L,d}&P_s\end{bmatrix}^{\top}$ is the input vector, $\epsilon_f(Z)$ is the bounded neural network approximation error, and $h(Z)=\begin{bmatrix}
h_1(Z)&\cdots&h_{n_1}(Z)
\end{bmatrix}^{\top}$ is the activation function, where $h_i(Z)$ is selected as
\begin{equation*}
h_i(Z)=\frac{(Z-\mu_i)^{\top}(Z-\mu_i)}{\upsilon_i^2}
\end{equation*}
for $i=1,\cdots,n_1$, where $\mu_i=\begin{bmatrix}
\mu_{i,1}&\mu_{i,2}&\mu_{i,3}
\end{bmatrix}^{\top}$ and $\upsilon_i$ are the center and width of the Gaussian transfer function.\\
The jump approximation basis function is defined as $\phi(Z)=\begin{bmatrix}
\phi_1(Z)&\cdots&\phi_{n_2}(Z)
\end{bmatrix}^{\top}$, where for $i=1,\cdots,n_2$
\begin{equation*}
\phi_i(Z)=\left\{ \begin{array}{ll}
0,&\text{for}\, Z<0,\\
(1-e_c^{-Z})^i,&\text{for}\,Z\geq 0\\
\end{array}\right.
\end{equation*}
where $e_c$ is the mathematical constant.\\
Therefore, the control variable \eqref{Eqn:lypu_d_L1_u} is replaced by
\begin{equation}\label{Eqn:low_u_nnn}
u=-\frac{1}{k_s}(k_3e_3+\hat W_1^{\top}h(Z)+\hat W_2^{\top}\phi(Z+c_1)+\hat W_3^{\top}\phi(Z+c_2))
\end{equation}
where $\hat W_i$ is the estimated values of $W_i$, for $i=1,2,3$.
Denoting by $\tilde W_i=W_i-\hat W_i$, $i=1,2,3$, the corresponding estimation errors, under input \eqref{Eqn:low_u_nnn}, consider the following Lyapunov candidate
\begin{equation*}
V_{L,2}=V_{L,1}+\frac{1}{2}\tilde W^{\top}\Gamma^{-1}\tilde W
\end{equation*}
where $\tilde W=\begin{bmatrix}
\tilde W_1^{\top}&\tilde W_2^{\top}&\tilde W_3^{\top}
\end{bmatrix}^{\top}$, $\Gamma=\text{diag}\{\Gamma_1,\Gamma_2,\Gamma_3\}$ is a positive definite matrix.
We also denote $\hat W=\begin{bmatrix}
\hat W_1^{\top}&\hat W_2^{\top}&\hat W_3^{\top}
\end{bmatrix}^{\top}$, $\chi(Z)=\begin{bmatrix}
h(Z)^{\top}&\phi(Z+c_1)^{\top}&\phi(Z+c_2)^{\top}
\end{bmatrix}^{\top}$.  Taking the derivative of $V_{L,2}$, it holds that
$$\begin{array}{ll}\dot V_{L,2}=\\
=-k_{\rho_3}e_3^2-\rho_3e_3\tau \dot x_v+\rho_3e_3(f_4-\hat W^{\top}\chi(Z))-\tilde W^{\top}\Gamma^{-1}\dot{\hat W}\\
=-k_{\rho_3}e_3^2-\rho_3e_3\tau \dot x_v+\rho_3e_3\epsilon_f(Z)+\rho_3e_3(\tilde W^{\top}\chi(Z)-
\tilde W^{\top}\Gamma^{-1}\dot{\hat W}\\
=-k_{\rho_3}e_3^2-\rho_3e_3\tau \dot x_v+\rho_3e_3\epsilon_f(Z)+\tilde W^{\top}(\rho_3e_3\chi(Z)-\Gamma^{-1}\dot{\hat W})
\end{array}$$
Therefore, the adaptation algorithm for $\hat W$ can be chosen as
\begin{equation}\label{Eqn: W_H}
\dot{\hat W}=\Gamma(\rho_3e_3\chi(Z)-\sigma\hat W)
\end{equation}
where $\sigma=\text{diag}\{\sigma_1,\sigma_2,\sigma_3\}$, is a robust matrix.\\
\subsection{Compensation for the control saturation}
Due to the control saturation described in \eqref{Eqn:saturation}, in view of the definition of $\delta$ in \eqref{Eqn:resigual_input}, from \eqref{Eqn:low_u_nnn}, the final control action applied to the low-level regulator is selected as
\begin{equation}\label{Eqn:low_u_nnn_final}
u=-\frac{1}{k_s}(k_3e_3+\hat W^{\top}\chi(Z))+\delta+\xi
\end{equation}
where $\xi$ is introduced to compensate for the control saturation effect and defined in the following auxillary system
\begin{equation}\label{Eqn:xi_aux}
\left\{ \begin{array}{ll}
\dot \xi=-k_{\xi}\xi-\frac{|\rho_3e_3k_s(\delta+\xi)|+0.5\delta^2}{\xi}-\delta,& |\xi|> \mu\\
\dot \xi=0,&|\xi|\leq \mu
\end{array}  \right.
\end{equation}
where $k_{\xi}>0$, and $\mu$ is a positive scalar.\\
In order to guarantee the robustness property with control action~\eqref{Eqn:low_u_nnn_final}, (see the following Section~\ref{sec:5}), the parameter $k_{\xi}$ is assumed to be selected greater than $\frac{1}{2}$.
\section{Properties of the closed-loop system}\label{sec:5}
In this section, the closed-loop robustness properties are discussed. To this end, in view of the cascade structure of the controller, the stability at the lower layer is considered first,  the one at the higher layer is then analyzed given the theoretical result at the lower layer.
The following standing assumption is concerned:
\begin{lemma}\label{Lemma1} \cite{selmic1997neural}
Let $f:\mathcal{X}\rightarrow R$ be any bounded function that is continuous and analytic on convex set $\mathcal{X}$ except at point $x=c$, then there exist a function
\begin{equation}\label{eqn:fx_appro}
\hat f(x)=g(x)+\sum_{i=0}^{T}a_i\phi_i(Z-c)
\end{equation}
such that
$$|f(x)-\hat f(x)|\leq \bar {\epsilon}$$
where $g(x)$ is a continuous function, $a_i$ is a scalar, and $\bar {\epsilon}$ is a positive scalar.
\end{lemma}
\begin{theorem}
	For piecewise discontinuous function $f_4$ defined in~\eqref{Eqn:low_m_3_e}, there exist a function $f_4$ of type
	\begin{equation*}
\hat f_4=W_1^{\top}h(Z)+W_2^{\top}\phi(Z+c_1)+W_3^{\top}\phi(Z+c_2)
	\end{equation*}
	such that
	\begin{equation}
	|\epsilon_f(Z)|\leq \bar {\epsilon}.
	\end{equation}
\end{theorem}
\begin{proof}
	In view of the result of Lemma~\ref{Lemma1}, substituting $g(x)$~\eqref{eqn:fx_appro} for continuous RBF function $W_1^{\top}h(Z)$ and extending discontinuous jump point $Z=c$ to multiple ones, e.g., $Z=-c_1,\, c_2$, leads to the result. \hfill $\square$
\end{proof}
The following result can be stated for the low-level controller:
\begin{theorem}\label{thm:l_s}   
	Under Assumptions~\ref{assu:measure} and \ref{rem:f_hat}, with the control action \eqref{Eqn:low_u_nnn_final} and the adaptation algorithm \eqref{Eqn: W_H}, the derivative of the Lyapunov function $V_{L,3}=V_{L,2}+\frac{1}{2}\xi^2$ converges to zero and  the variables $e_3$ and $\tilde{W}_i$, $i=1,2,3$, are uniformly ultimately bounded within the set, i.e.,
	\begin{subequations}\label{Eqn:state_b}
		\begin{align} |e_3|\leq& \sqrt{\frac{\alpha_1}{k_{30}}}\label{Eqn:e_3-b}\\
 \|\tilde W_1\|\leq& \sqrt{\frac{2\alpha_1}{\sigma_1}}		\label{Eqn:W_f_t_b}\\
\|\tilde W_2\|\leq& \sqrt{\frac{2\alpha_1}{\sigma_2}}\label{Eqn:W_2_t_b} \\
 \|\tilde W_3\|\leq& \sqrt{\frac{2\alpha_1}{\sigma_3}}\label{Eqn:W_3_t_b}
		\end{align}
	\end{subequations}
	
	where  $\alpha_1$ is defined in \eqref{Eqn:VL3_B}, and $\|\cdot\|$ is the Euclidean norm.
\end{theorem}
\begin{proof}
Along the similar line with Theorem~\ref{theo:1}, the monotonicity of the Lyapunov function $V_{L,3}$ has to be checked for all the Filippov solutions, i.e., the continuous sub-regions (where $|\xi|>\mu$, or $|\xi|<\mu$) and the separating surfaces (where $|\xi|=\mu$).\\
In view of this and of the analysis in~\cite{di2016extended}, with~\eqref{Eqn:low_u_nnn_final} and~\eqref{Eqn:xi_aux}, one can write the derivative of $V_{L,3}$ in the form
$$\begin{array}{ll}
\dot V_{L,3}=
-k_{\rho_3}e_3^2-
f_{\xi,1}^{\lambda}+f_{\xi,2}^{\lambda}-\rho_3e_3\tau \dot x_v+\rho_3e_3\epsilon_f(Z)+\tilde W^{\top}\sigma \hat{W}
\end{array}$$
where for $i=1,2$, $$f_{\xi,i}^{\lambda}=\left\{\begin{array}{ll}f_{\xi,i}&|\xi|>\mu\\
\lambda f_{\xi,i}& |\xi|=\mu\\
0&|\xi|<\mu,
\end{array}\right.$$
$\lambda\in [0\ 1]$, $f_{\xi,1}=(k_{\xi}-\frac{1}{2})\xi^2$, $f_{\xi,2}=\rho_3e_3k_s(\delta+\xi)-|\rho_3e_3k_s(\delta+\xi)|$.\\
Since the term $\rho_3e_3k_s(\delta+\xi)-|\rho_3e_3k_s(\delta+\xi)|\leq0$, one has
	\begin{equation}\label{Eqn:dlyap_l3}
	\begin{array}{ll}
	\dot V_{L,3}\leq&-k_{\rho_3}e_3^2-f_{\xi,1}^{\lambda}-\rho_3e_3\tau \dot x_v+\\
	&\rho_3e_3\epsilon_f(Z)+\tilde W^{\top} \sigma\hat{W}\\
	\end{array}
	\end{equation}
	Define $k_{30}>0$, $k_{31}>0$, and $k_{32}>0$, such that $k_{ 30}+k_{31}+k_{32}=k_{\rho_3}$, one can prove
	\begin{subequations}\label{Eqn:lyap_k31k33e3}
	\begin{align}\label{Eqn:lyap_k31e3}
	-k_{31}e_3^2-\rho_3e_3\tau \dot x_v\leq\frac{(\rho_3\tau\dot x_v)^2}{4k_{ 31}}
	\end{align}

	\begin{align}\label{Eqn:lyap_k33e3}
	-k_{32}e_3^2+\rho_3e_3\epsilon_f(Z)\leq\frac{(\rho_3\bar {\epsilon})^2}{4k_{32}}
	\end{align}
		\end{subequations}
	Recalling the input constraint $|u|\leq u_{max}$, taking integral of both side of ${{\dot x}_v} = \frac{1}{\tau }({k_s}u - {x_v})$, it holds that
	$$
	\begin{array}{ll}
	|x_v(t)|&=|x_v(0)e_c^{-\frac{t}{\tau}}+\frac{k_s}{\tau}\int^{\top}e_c^{-\frac{1}{\tau}(t-\sigma)}ud\sigma|\\
	&\leq |x_v(0)|e_c^{-\frac{t}{\tau}}+\frac{k_s}{\tau}u_{max}|\int^{\top}e_c^{-\frac{1}{\tau}(t-\sigma)}d\sigma|:=\rho
	\end{array}$$
	Then one can also prove that $|\dot x_v|\leq \frac{1}{\tau}(k_su_{max}+\rho):=\varsigma$. In view of this, and recalling \eqref{Eqn:lyap_k31k33e3}, from \eqref{Eqn:dlyap_l3}, one can compute
	\begin{equation}\label{Eqn:lypu_d_L11}
	\begin{array}{ll}
	\dot V_{L,3}
	\leq&-k_{\rho_3}e_3^2-f_{\xi,1}^{\lambda}+\\
	&\frac{(\rho_3\tau \varsigma)^2}{4k_{ 31}}+\frac{(\rho_3\bar {\epsilon})^2}{4k_{32}}+\tilde W^{\top}\sigma \hat{W}\\
	\end{array}
	\end{equation}
	Considering that, for $i=1,2,3$
	\begin{equation*}\label{Eqn:W_f_b}
	\begin{array}{ll}
	\sigma_i\tilde W_i^{\top} \hat{W}_i&=\sigma_i\tilde W_i^{\top} (W_i-\tilde{W}_i)\\
	&\leq-\frac{1}{2}\sigma_i\|\tilde W_i\|^2+\frac{1}{2}\sigma_i\|W_i\|^2
	\end{array}
	\end{equation*}
	Then
	\begin{equation}\label{Eqn:VL3_B}
	\begin{array}{ll}
	\dot V_{L,3}\leq&-k_{30}e_3^2-f_{\xi,1}^{\lambda}-\frac{1}{2}\sigma_1\|\tilde W_1\|^2-\\&\frac{1}{2}\sigma_2\|\tilde W_2\|^2-\frac{1}{2}\sigma_3\|\tilde W_3\|^2+\alpha_1
	\end{array}
	\end{equation}
	where $\alpha_1=\frac{1}{2}\sigma_1\|W_1\|^2+\frac{1}{2}\sigma_2\|W_2\|^2+\frac{1}{2}\sigma_3\|W_3\|^2+\frac{(\rho_3\tau \varsigma)^2}{4k_{ 31}}+\frac{ ({\rho_3\bar\epsilon})^2}{4k_{32}}$.
	In view of \eqref{Eqn:VL3_B},  $\dot V_{L,3}$ converges to zero. Considering also that $-\dot V_{L,3}\geq k_{30}e_3^2+f_{\xi,1}^{\lambda}+\frac{1}{2}\sigma_1\|\tilde W_1\|^2+\frac{1}{2}\sigma_2\|\tilde W_2\|^2+\frac{1}{2}\sigma_3\|\tilde W_3\|^2-\alpha_1$,  result \eqref{Eqn:state_b} follows, see~\cite{di2013hybrid}.
\hfill $\square$
\end{proof}
Assume now that the low-level controller has been run such that the condition \eqref{Eqn:e_3-b} has been achieved. From \eqref{Eqn:e_3-b}, it can be noted that the actual input to be applied to the system \eqref{Eqn:e_3-b-model} is not exactly $F_{L,d}$ but $F_{L,d}$ plus a residual term, i.e.,
\begin{equation}\label{Eqn:u_h-d}
F_L=F_{L,d}+\eta
\end{equation}
where $\|\eta\|\leq\sqrt{\frac{\alpha_1}{k_{\rho_3}}}+\kappa $, $\kappa$ is null or a small positive value due to~\eqref{Eqn:F_f_d}. \\
The following result can be stated for the high-level controller:

\begin{theorem}   
	Under the result of Theorem~\ref{thm:l_s}, with control  \eqref{Eqn:u_h-d}, and adaptation algorithms \eqref{Eqn:esti_update} and \eqref{Eqn:fric_para_adapt}, the derivative of the Lyapunov function $V_{H,3}$ converges to zero and the variables $e_1$, $e_2$ are uniformly ultimately bounded, i.e.,
	\begin{subequations}\label{Eqn:state_h_b}
		\begin{align*} |e_1|\leq& \sqrt{\frac{M^2\alpha_1}{4k_1k_{\rho_3}k_{21}}}\\
		 |e_2|\leq& \sqrt{\frac{M^2\alpha_1}{4k_{2}k_{\rho_3}k_{21}}}
		\end{align*}
		
	\end{subequations}

\end{theorem}

\begin{proof}
Along the similar line with Theorem~\ref{theo:1}, for all the Filippov solutions, applying the control input \eqref{Eqn:u_h-d}, the derivative of Lyapunov function $V_{H,3}$ can be computed
\begin{equation*}
\dot V_{H,3}=-k_1e_1^2-k_2e_2^2+Me_2\eta.
\end{equation*}	
Denoting $k_{20}>0$, $k_{21}>0$ such that $k_2=k_{20}+k_{21}$, and following the same line with \eqref{Eqn:lyap_k31k33e3}, one can also prove that

\begin{equation*}
-k_{21}e_2^2+Me_2\eta\leq \frac{M^2\alpha_1}{4k_{\rho_3}k_{21}}
\end{equation*}
Then, it holds that
\begin{equation*}
\dot V_{H,3}\leq-k_1e_1^2-k_2e_2^2+\frac{M^2\alpha_1}{4k_{\rho_3}k_{21}}
\end{equation*}
\hfill $\square$
	\end{proof}

\begin{remark}
If  the parameter $\sigma$ in \eqref{Eqn: W_H} is chosen sufficiently small and the neural network has been tuned properly (by increasing the number of neurons) such that the bound $\bar \epsilon$ is made arbitrarily small, then the Lyapunov function $V_{L,3}$, $e_3$ and $e_4$ converge to zero, and subsequently, $V_{H,3}$, $e_1$, and $e_2$ converge to zero as well.
\end{remark}

\section{Simulation results}\label{sec:6}
\begin{figure}[ht]
	\center
	\includegraphics[width=0.95\columnwidth]{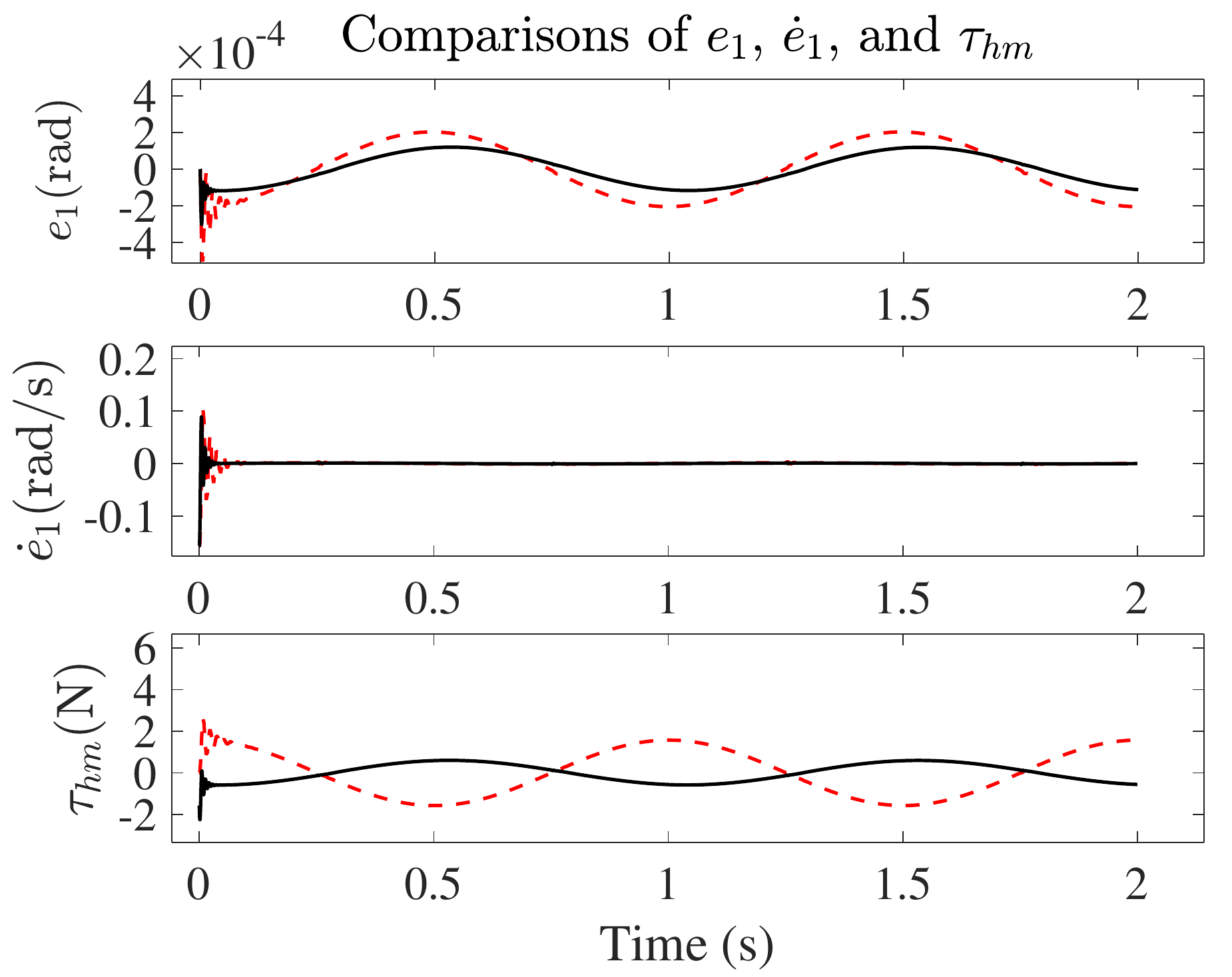}
	\caption{Comparison of $e_1$, $\dot e_1$, and $\tau_{hm}$: black solid lines represent the values computed with the proposed approach, while red dashed lines stand for the ones computed with PD controller.}
	\label{fig:states}
\end{figure}
\begin{figure}[h]
	\center
	\includegraphics[width=0.95\columnwidth]{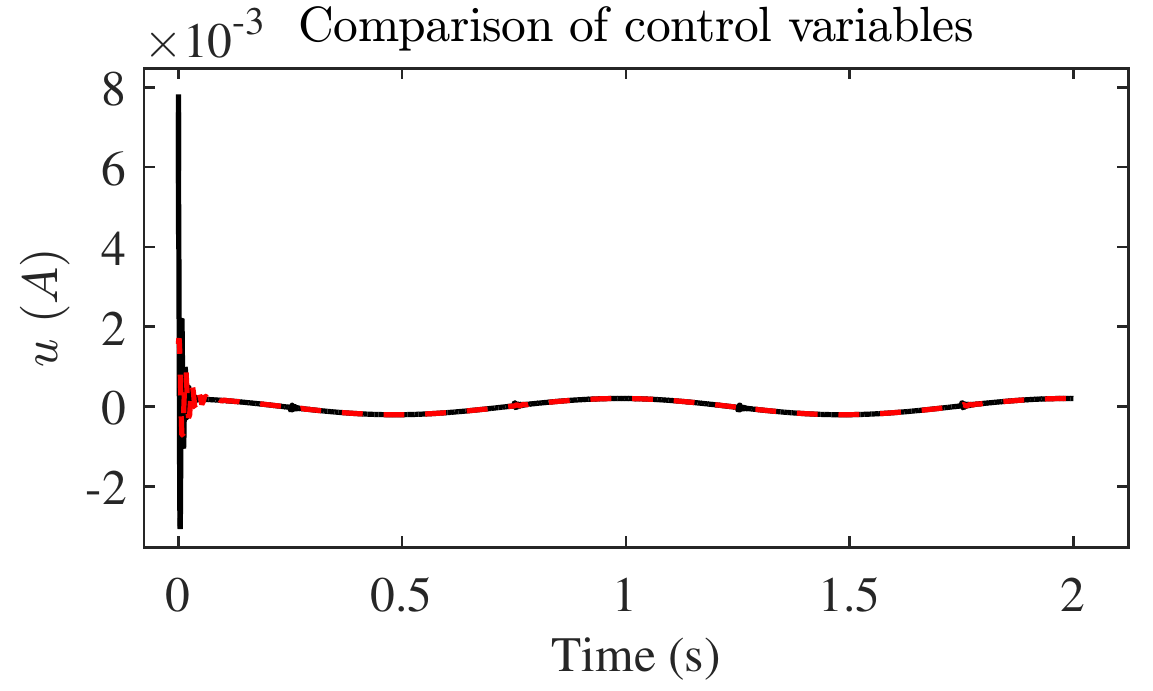}
	\caption{Comparison of control variable: black solid line is the values computed with the proposed approach, while red dashed line represent the one computed with PD controller..}
	\label{fig:u}
\end{figure}
In this section, simulation results
are reported to show the performance of the hierarchical adaptive control algorithms previously described. The values of the key parameters of system \eqref{Eqn:state} are listed in Table~\ref{tb:para}.
\begin{table}[hb]
\begin{center}
\centering\caption{Model parameters}\label{tb:para}
\vskip 0.2cm
\scalebox{0.8}{
\begin{tabular}{cccc}
Parameters & Values & Parameters & Values \\\hline
$m$ ($kg$) & 70 & $C_{in}+C_{ex}$  ($m^3/s \cdot Pa$)  & $2\times 10^{-14}$\\ \hline
$P_p $ ($MPa$) & 5 & $k_s$ (m/A) & 0.0146 \\ \hline
$b$ ($N\cdot s/m$) & 5000 & $A_1$ ($m^2$) & $3.25\times  10^{-4}$ \\ \hline
$F_C$ ($N$) & 8 & $A_2$ ($m^2$) & $2.10\time 10^{-4} $\\ \hline
$\tau$ ($N\cdot s/m$) & 0.0015 & $l_0 ($m$)$ & 0.28 \\ \hline
$K_c$ ($m^3s/Pa$) & $8.8\times 10^{-16}$ & $K_q$ ($m^3s\cdot A$) & 0.52 \\ \hline
$x_{c0}(m)$  & 0.07 & $l_0$ ($m$) & 0.1 \\ \hline
\end{tabular}
}
\end{center}

\end{table}

\begin{figure}[h!]
	\center
	\includegraphics[width=0.95\columnwidth]{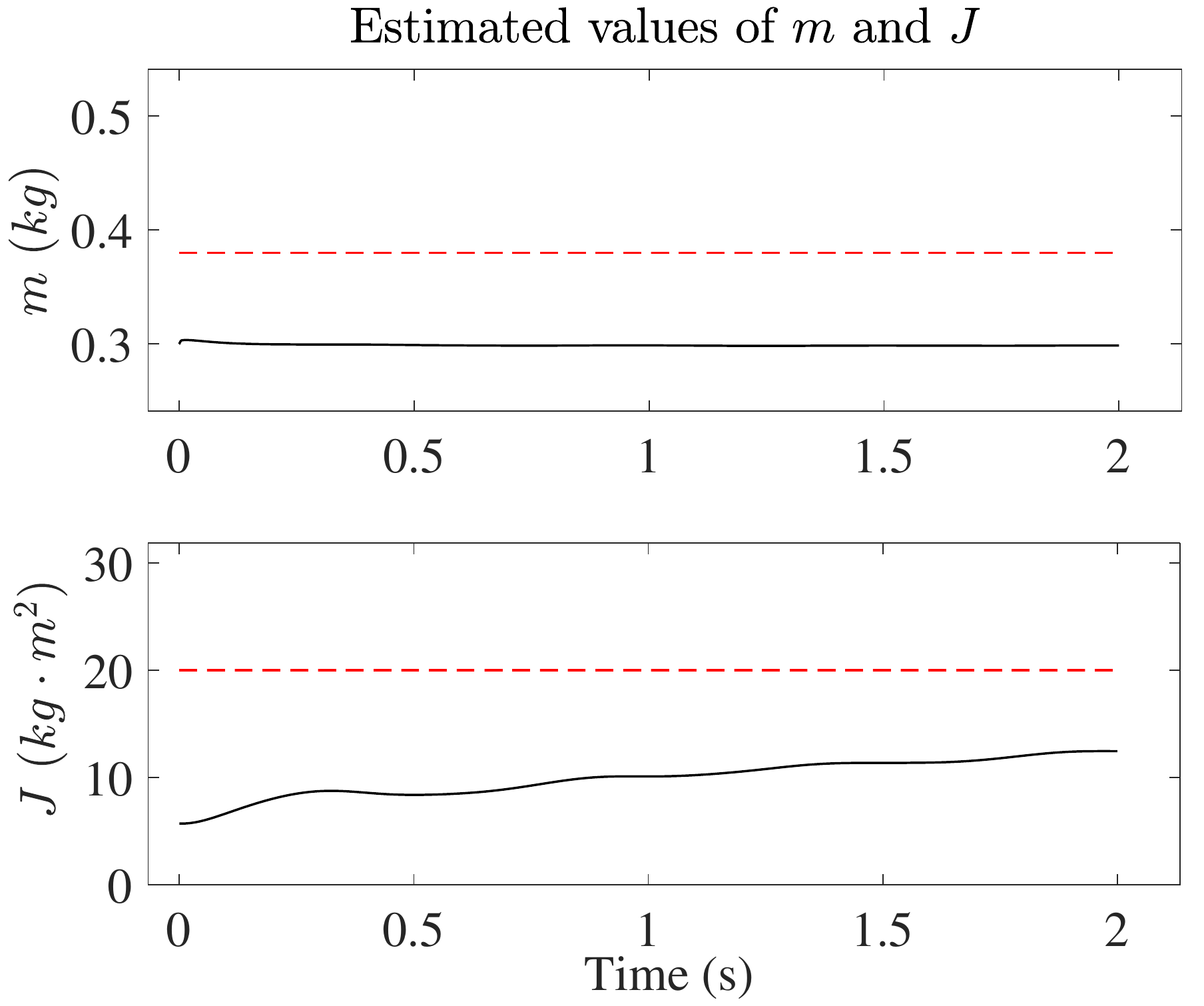}
	\caption{Estimated values of $m$ and $J$: black solid lines are the  estimated values, while red dashed lines are the true ones.}
	\label{fig:mass_inertia}
\end{figure}
\begin{figure}[h!]
	\center
	\includegraphics[width=0.95\columnwidth]{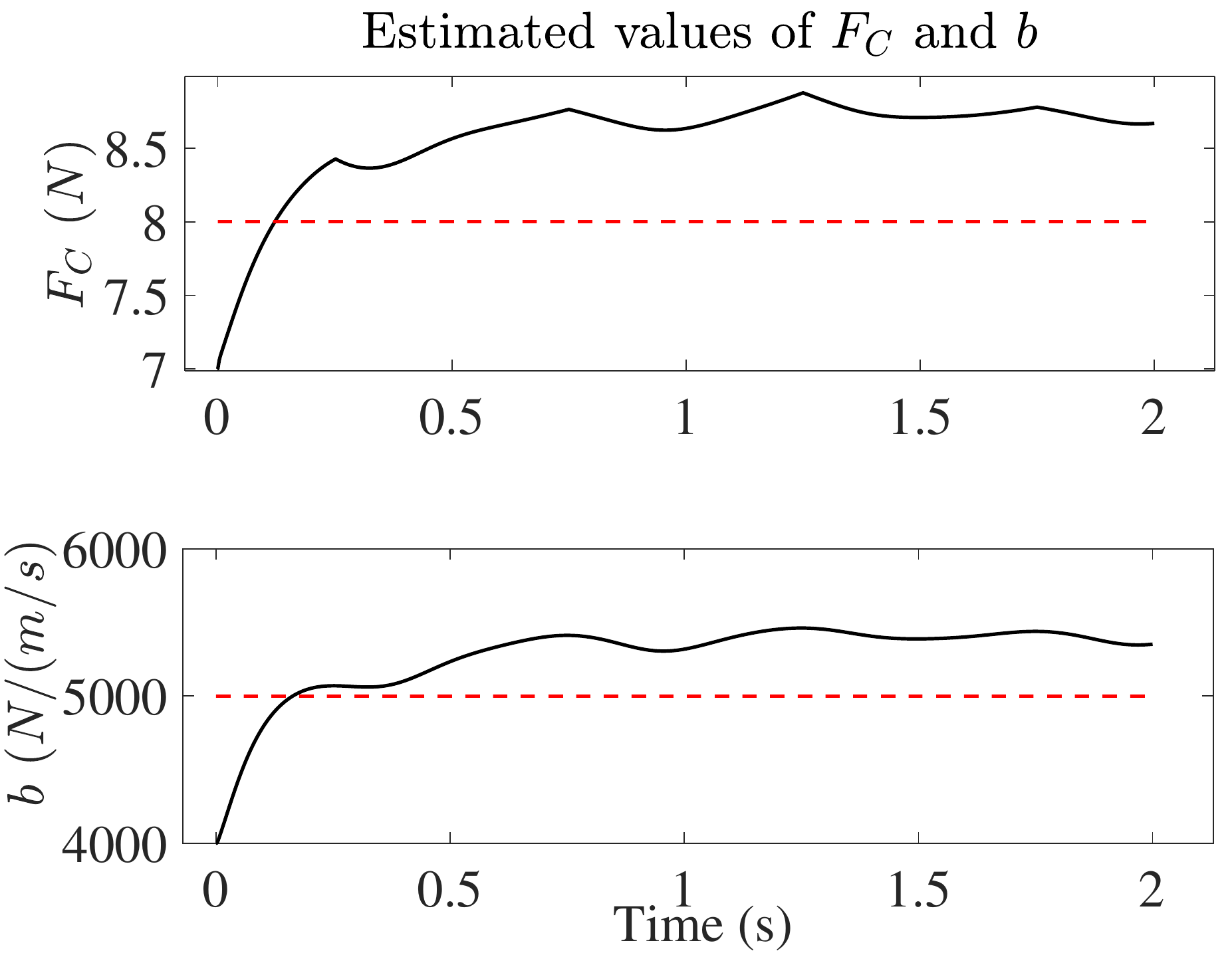}
	\caption{Estimated values of $F_c$ and $b$: black solid lines represent the estimated values, while red dashed lines stand for  the true ones.}
	\label{fig:FC_b}
\end{figure}
\begin{figure}[h!]
	\center
	\includegraphics[width=0.95\columnwidth]{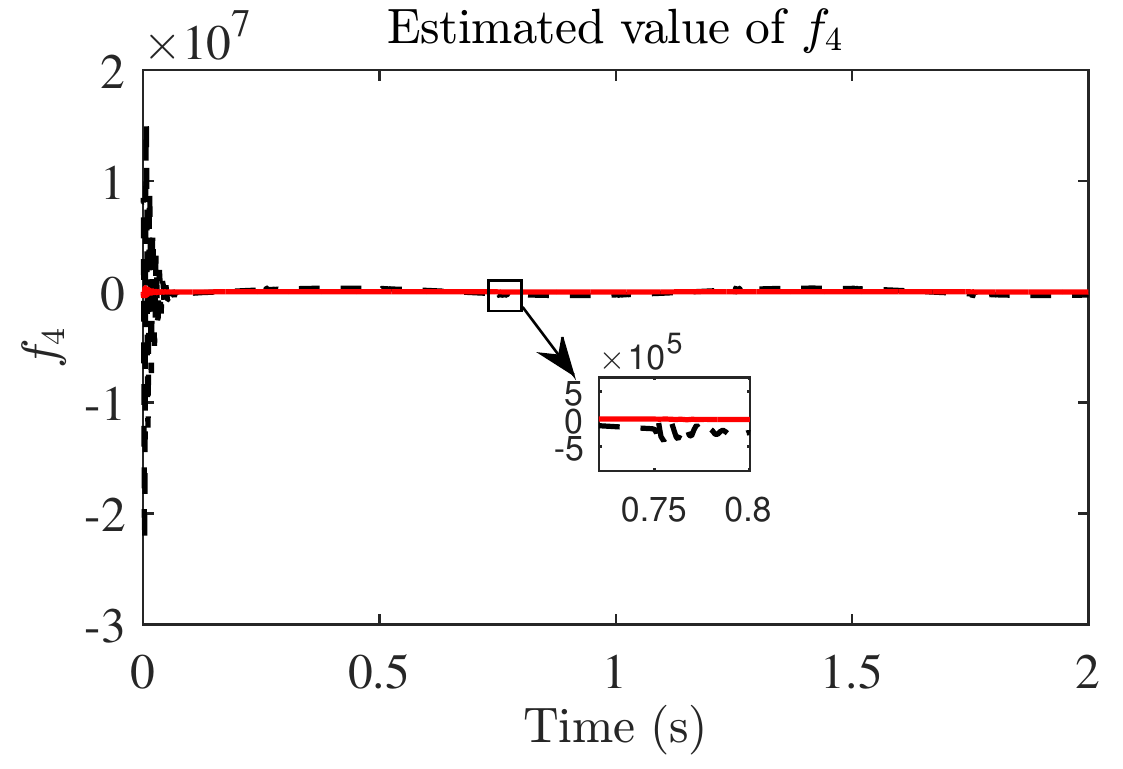}
	\caption{Estimated values of $f_4$: black solid lines is the  estimated value, while red dashed line is the true one.}
	\label{fig:W_n5}
\end{figure}
In the simulation experiment, the value of the reference joint angle has been set to $\phi_d=0.025sin(2\pi t)$ and the sampling time has been selected as $0.001$ $s$. The maximal value of the control variable is $u_{max}=2.5\times 10^{-2}$ $A$. The parameters in \eqref{Eqn:tau_hm} are $k_{\varphi}=5000$ and $k_{\dot \varphi}=10$. The high-level regulator has been designed with $k_1=500$, $k_2=200$, and $\rho_1=\rho_2=1$. The corresponding adaptation algorithm has been devised with $q_J=1000$, $q_m=0.01$, $q_1=0.007$, and $q_2=0.0005$. The low-level controller has been implemented with  $k_3=1000$, $\rho_3=1$. The parameter in~\eqref{Eqn:F_f_d} is selected as $\epsilon_0=0.001$, In neural network \eqref{Eqn:neural_network}, the continuous RBF contains $n_1=3^3$ nodes with centers of $Z$ evenly spaced in the domain $[-4,\,4]\times[-1,\,1]\times[-4,\,4]$, and the values of the variance is chosen as $\upsilon=20000$. Also the nodes for the jump approximation function centered at $c_1$ and $c_2$ are $n_2=8\times 4$. The  parameter for the updating law \eqref{Eqn: W_H} is $\Gamma=100000 I_{n_1+n_2+n_3}$ and the robust term is $\sigma=0.2I_{n_1+n_2+n_3}$. The simulation experiment has been run with null initial conditions. For comparison, a PD controller is designed with proportional gain $k_P=-1$ and derivative gain $k_D=-0.01$. The evolution of the states and control variables with the proposed approach and the PD algorithm, are  reported in Figure.~\ref{fig:states}-\ref{fig:u}. From Figure.~\ref{fig:states}, it can be seen that after an initial transient, the proposed control algorithm shows satisfactory tracking performance, while the tracking error and the interaction torque computed with the proposed algorithm are smaller than that with the PD algorithm.  Also, the estimations of the uncertain parameters and the nonlinear function $f_4$ of the system with the proposed algorithm are  reported in Fig.~\ref{fig:mass_inertia}-\ref{fig:W_n5},
which shows that, the estimated values are close to their true values. Note that, some of the estimated values are not converging to their actual value. To further show the convergence properties of the estimated values, their true values must be known a-priori, which is usually nontrivial for practical reasons.
\section{Conclusion}\label{sec:7}
In this paper, a hierarchical Lyapunov-based adaptive cascade control scheme of a lower-limb exoskeleton with control saturation has been developed for joint angle position tracking objective.  Adaptation algorithms have been proposed to estimate unknown model parameters at the both layers. At the lower layer, the neural network with continuous and discontinuous basis function has been used to approximate piecewise discontinuous nonlinear function. Thanks to the estimating techniques prescribed, the proposed approach allows to use imprecise models, which is much more reasonable for practical reasons. Moreover,  with suitable control parameters design, this approach can also minimize the interaction torque between machine and human. The robustness of the  closed-loop system has been discussed under control saturation. Simulation experiments including a comparison with PD have been reported, showing that the proposed approach is satisfactory in tracking performance and in interaction torque reduction, and outperforms PD controller in these respects.
Future work will consider the use of learning based algorithms at the lower layer so as to optimize the switching condition between Mode1 and Mode2.

\section*{Acknowledgements}

This work is supported by the National Key R$\&$D Program of China
2018YFB1305105 and the National Natural Science Foundation of China
under Grants U1564214, 61751311,  618\break25305.

\bibliographystyle{plain}
\bibliography{ifacconf}             

\end{document}